\documentclass[sigconf]{acmart}

\usepackage{amsmath, amsthm, color}

\pdfoutput=1

\usepackage{tikz-cd}
\usepackage{graphicx}
\usepackage{caption}
\usepackage{mathtools}
\usepackage{enumitem}
\usepackage{xfrac}
\usepackage{verbatim}
\usepackage{tikz}
\usetikzlibrary{matrix}
\usetikzlibrary{arrows}
\usepackage[linesnumbered,ruled,vlined]{algorithm2e}
\usepackage[noend]{algpseudocode}
\usepackage{caption}
\usepackage{subcaption}
\usepackage{dutchcal}

\oddsidemargin \evensidemargin
\usepackage{makecell}
\usepackage{array}

\newtheorem{theorem}{Theorem}
\newtheorem{proposition}[theorem]{Proposition}
\newtheorem{lemma}[theorem]{Lemma}
\newtheorem{corollary}[theorem]{Corollary}

\theoremstyle{definition}
\newtheorem{definition}[theorem]{Definition}

\DeclareMathOperator{\rank}{rank}
\DeclareMathOperator{\Hom}{Hom}
\DeclareMathOperator{\id}{id}
\DeclareMathOperator{\I}{I}

\title{\bf A quasi-optimal lower bound for skew polynomial multiplication}
\date{}

\copyrightyear{2024}
\acmYear{2024}
\setcopyright{rightsretained}
\acmConference[ISSAC 2024]{International Symposium on Symbolic and Algebraic Computation 2024}{July 16--19, 2024}{Raleigh, USA}
\acmBooktitle{International Symposium on Symbolic and Algebraic Computation 2024 (ISSAC 2024), July 16--19, 2024, Raleigh, USA}\acmDOI{XXXXXXX.XXXXXXX}
\acmISBN{XXXXXXX.XXXXXXX}

\begin{document}
\author{Qiyuan Chen$^{a,b}$}
		\affiliation{%
			\institution{$^{a}$KLMM, Academy of Mathematics and Systems Science, Chinese Academy of Sciences}
			\institution{$^{b}$School of Mathematical Sciences, University of Chinese Academy of Sciences}
			\city{Beijing}
			\country{China}}
		\email{chenqiyuan@amss.ac.cn}

\author{Ke Ye$^{a,b}$}
		\authornote{Corresponding author.}
		\affiliation{%
			\institution{$^{a}$KLMM, Academy of Mathematics and Systems Science, Chinese Academy of Sciences}
			\institution{$^{b}$School of Mathematical Sciences, University of Chinese Academy of Sciences}
			\city{Beijing}
			\country{China}}
		\email{keyk@amss.ac.cn}

\begin{abstract}
We establish a lower bound for the complexity of multiplying two skew polynomials. The lower bound coincides with the upper bound conjectured by Caruso and Borgne in 2017, up to a log factor. We present algorithms for three special cases, indicating that the aforementioned lower bound is quasi-optimal. In fact, our lower bound is quasi-optimal in the sense of bilinear complexity. In addition, we discuss the average bilinear complexity of simultaneous multiplication of skew polynomials and the complexity of skew polynomial multiplication in the case of towers of extensions.
\end{abstract}




\maketitle
\section{Introduction}
The skew polynomial ring is a non-commutative analogue of the usual polynomial ring. It is the special case of the Ore algebra first studied in \cite{ore1933theory}. Because of its highly non-trivial algebraic and computational properties, the ring of skew polynomials plays a crucial role in diverse fields of mathematics. For instance, quaternion algebras and cyclic algebras are quotients of skew polynomial rings \cite{Knus98,GS17}; the connection between matrix algebras and skew polynomial rings can be used to design fast algorithms for matrix multiplication \cite{QKX24}; it is just realized in recent years that skew polynomial rings over finite fields provide us new models in coding theory \cite{SK09,DF09,BU14}. Due to aforementionedi applications in complexity theory and coding theory, various operations of skew polynomials are extensively studied in the literature as well, including multiplication \cite{Giesbrecht98, 2017Fast, GHS20}, factorization \cite{Giesbrecht98,XJ17}, Gr\"{o}bner bases \cite{RV13} and interpolation  \cite{SA16,SA18}.

As an algebra, the most important and fundamental operation on skew polynomial rings is the multiplication. We recall that given a $\mathcal{k}$-algebra $\mathcal{A}$ of dimension $r$ and a $\mathcal{k}$-linear automorphism $\sigma$ of $\mathcal{A}$, the skew polynomial ring $\mathcal{A}[x,\sigma]$ consists of polynomials with coefficients in $\mathcal{A}$, whose multiplication is skewed by $\sigma$, i.e., $x a = \sigma (a) x, a\in \mathcal{A}$. We denote by $C_{\mathcal{k}} (\mu_d)$ the number of arithmetic operations over $\mathcal{k}$ required to compute the product of two degree $d$ skew polynomials in $\mathcal{A}[x,\sigma]$. The first fast algorithm of skew polynomial multiplication is proposed in \cite{Giesbrecht98}, which has complexity $C_{\mathcal{k}} (\mu_d) = O(dr^2 + d^2 r)$. Algorithms presented in \cite{SA16,SA18} improve the upper bound to $\widetilde{O}(d^{(\omega + 1)/2} r)$, where $\omega$ denotes the exponent of matrix multiplication. Based on the quasi-optimal bound \cite{BBv12} for the multiplication of linear differential operators, it is conjectured in \cite{2017Fast} that $C_{\mathcal{k}} (\mu_d)=\widetilde{O}(d\min(d,r)^{\omega-2}r)$ . An upper bound is also obtained in \cite{2017Fast}:
\begin{equation}\label{eqn:2017bound}
C_{\mathcal{k}} (\mu_d) =
\begin{cases}
\widetilde{O}(d r^{\omega-1}),\quad d\ge r\\
\widetilde{O}(d^{\omega-2} r^{2}),\quad d< r
\end{cases},
\end{equation}
which coincides with the conjectured upper bound when $d \ge r$. As far as we aware, it is the best upper bound in the literature when $d \ge r^{2/(5 -\omega)}$. However, if $d < r^{2/(5 -\omega)}$, the bound $\widetilde{O}(d^{(\omega + 1)/2} r)$ in \cite{SA16,SA18} is better. By exploiting the mod-$r$ sparsity $R\le r$ of the support of the product, \cite{GHS20} proposes a Las Vegas algorithm of complexity $\widetilde{O}(\max(d,r) rR^{\omega - 2})$, which outperforms existing algorithms if $d \ge  \min(r^{2/(5 - \omega)},r^{2/(\omega + 1)}R^{(2\omega- 4)/(\omega + 1)})$. Lastly, we remark that although the upper bound of skew polynomial multiplication has been studied extensively in the past two decades, the lower bound is still far from being understood. 
\subsection*{Contributions}
This paper is concerned with the computational complexity of skew polynomial multiplication. The primary goal is twofold: 
\begin{enumerate}
\item we establish a lower bound of $C_{\mathcal{k}} (\mu_d)$ in Theorem~\ref{thm:lower bound}, which coincides with the conjectured upper bound $\widetilde{O}(d\min(d,r)^{\omega-2}r)$ in \cite{2017Fast}, up to a log factor;
\item we present in Section~\ref{sec:algorithms} fast algorithms for low degree ($d \le r$) skew polynomial multiplication in several special cases, which cost $\widetilde{O}(d^{\omega-1}r)$ arithmetic operations. This indicates that our lower bound is quasi-optimal.
\end{enumerate}
In particular, our lower bound together with the algorithm presented in  \cite{2017Fast} implies that if $d \ge r$, then $C_{\mathcal{k}} (\mu_d)$ is completely determined (up to a log factor):
\[
C_{\mathcal{k}} (\mu_d) = \widetilde{O}(d r^{\omega - 1}) = \Omega (d r^{\omega - 1}).
\]
Additionally, in Proposition~\ref{prop:average}, we establish an upper bound of the average bilinear complexity of simultaneously multiplying several pairs of low degree skew polynomials.
\section{Preliminaries}
In this section, we record some notations, definitions and basic facts from complexity theory and algebra, which are necessary for the rest of this paper.
\subsection{Notations for complexity}
Given functions $f,g:\mathbb{N} \to \mathbb{N}$, we denote $f(n) = O(g(n))$ if there exists some constant $C > 0$ such that 
\[
f(n) \le C g(n)
\] 
for sufficiently large $n$. We denote $f(n) = \widetilde{O}(g(n))$ if there exists some constant $C,\tau > 0$ such that 
\[
f(n) \le C (\log(n))^\tau g(n)
\] 
for sufficiently large $n$. Moreover, we write $f(n) = \Omega(g(n))$ if there exists some constant $C > 0$ such that 
\[
f(n) \ge C g(n)
\] 
for sufficiently large $n$. 
\subsection{Bilinear complexity}
For convenience, we collect some basic facts about the bilinear complexity. The standard reference for this subsection is \cite{burgisser2013algebraic}.
\begin{definition}\label{def2}
Let $\mathcal{R}$ be a commutative ring and let $\mathbb{U},\mathbb{V},\mathbb{W}$ be finitely generated free $\mathcal{R}$-modules. The bilinear complexity (or rank) of a $\mathcal{R}$-bilinear map $f: \mathbb{U} \times \mathbb{V} \to \mathbb{W}$, denote by $\rank_{\mathcal{R}}(f)$, is the smallest positive integer $r$ to ensure the existence of $\alpha_j \in \Hom (\mathbb{U},\mathcal{R}),\beta_j \in \Hom (\mathbb{V},\mathcal{R})$ and $w_j \in \mathbb{W}$, $1 \le j \le r$, such that
\[
f (u,v) = \sum_{j=1}^r \alpha_j (u) \beta_j(v) w_j, \quad (u,v)\in \mathbb{U}\times \mathbb{V}.
\]
\end{definition}
We denote by $C_{\mathcal{R}}(f)$ the total number of arithmetic operations required to compute $f$ over $\mathcal{R}$. In the literature, $C_{\mathcal{R}}(f)$ is called the \emph{total complexity} of $f$. It is obvious that
\begin{equation}\label{eqn:rank<complexity}
\rank_{\mathcal{R}}(f) \le C_{\mathcal{R}}(f).
\end{equation}
Let $\mathcal{S}$ be a commutative ring containing $\mathcal{R}$ as a sub-ring. For each free $\mathcal{R}$-module $\mathbb{U}$, we denote $\mathbb{U}^{\mathcal{S}} \coloneqq \mathbb{U} \otimes_{\mathcal{R}} \mathcal{S}$. Similarly, if $f:\mathbb{U}\times \mathbb{V} \to \mathbb{W}$ is a $\mathcal{R}$-bilinear map, then we denote by $f^{\mathcal{S}}: \mathbb{U}^{\mathcal{S}} \times \mathbb{V}^{\mathcal{S}} \to \mathbb{W}^{\mathcal{S}}$ the $\mathcal{S}$-bilinear map obtained by extending $f$ naturally. By definition, we have 
\begin{equation}\label{eqn:rank extn}
\rank_{\mathcal{S}} (f^{\mathcal{S}}) \le \rank_{\mathcal{R}} (f).
\end{equation}
If there exist finitely generated free $\mathcal{R}$-modules $\mathbb{U}', \mathbb{V}', \mathbb{W}'$, $\mathcal{R}$-bilinear map $f': \mathbb{U}' \times  \mathbb{V}' \to \mathbb{W}'$ and $\mathcal{R}$-linear maps $\varphi_1:\mathbb{U} \to \mathbb{U}'$, $\varphi_2:\mathbb{V} \to \mathbb{V}'$ and $\varphi_3:\mathbb{W}' \to \mathbb{W}$ such that $f = \varphi_3 \circ f' \circ (\varphi_1 \times \varphi_2)$, then we say that \emph{$f$ is a restriction of $f'$}, denoted by $f \le f'$. Clearly, $f \le f'$ implies
 \begin{equation}\label{eqn:rank restriction}
\rank_{\mathcal{R}} (f) \le \rank_{\mathcal{R}} (f').
\end{equation}
%
\subsection{Exponent of matrix multiplication}
Let $\mathcal{R}$ be a commutative ring. We denote by $\langle m,n,p\rangle_{\mathcal{R}}$ the $\mathcal{R}$-bilinear map of multiplying an $m\times n$ matrix with an $n\times p$ matrix over $\mathcal{R}$. The \emph{exponent of matrix multiplication} over $\mathcal{R}$ is \[
\omega(\mathcal{R}) \coloneqq \inf \lbrace 
\tau \in \mathbb{R}: C_{\mathcal{R}} (\langle n,n,n\rangle_{\mathcal{R}}) = O(n^\tau) 
\rbrace.\]
The same proof of Corollary~15.18 in \cite{burgisser2013algebraic} leads to:
\begin{lemma}\label{lem:rank matrix extn}
For any commutative $\mathcal{k}$-algebra $\mathcal{R}$, $\omega(\mathcal{R}) = \omega(\mathcal{k})$.
\end{lemma}
Because of Lemma~\ref{lem:rank matrix extn}, we simply abbreviate $\omega(\mathcal{R})$ by $\omega$. The proof of Theorem~15.11 in \cite{burgisser2013algebraic} can be extended to show:
\begin{theorem}\label{thm5}
If there exist positive integers $e_{i},h_{i},l_{i}$, $1\le i \le s$, such that $\rank_{\mathcal{R}} \left(\bigoplus\limits_{i=1}^{s}\langle e_{i},h_{i},l_{i}\rangle_{\mathcal{R}}\right)\le \tau$, then $\sum\limits_{i=1}^s (e_{i}h_{i}l_{i})^{\omega/3}\le\tau$.
 \end{theorem}
\subsection{\'{E}tale algebra and Galois algebra}\label{subsec:etale algebra}
Let $\mathcal{k}$ be a field and let $\mathcal{A}$ be a finite \'{e}tale $\mathcal{k}$-algebra, i.e., $\mathcal{A}$ is a finite product of finite separable field extensions of $\mathcal{k}$. 
\begin{theorem}[primitive element theorem]\cite[Proposition~4.1]{UZ17}\label{thm:primitive element thm}
If $\mathcal{k}$ is an infinite field and $\mathcal{A}$ is a finite \'{e}tale $\mathcal{k}$-algebra, then there exists $a\in \mathcal{A}$ such that $\mathcal{A} = \mathcal{k}[a]$.
\end{theorem}
Assume further that $\sigma$ is an automorphism of $\mathcal{A}$ such that $\mathcal{A}^{\sigma} = \mathcal{k}$ and the cyclic group $\langle \sigma \rangle$ generated by $\sigma$ has order $r = \dim_{\mathcal{k}} \mathcal{A} > 1$. We say that $\mathcal{A}$ is a \emph{$\langle \sigma \rangle$-Galois algebra} \cite[Section~18.~B]{Knus98}. Examples of $\langle \sigma \rangle$-Galois algebras include:
\begin{itemize}
\item \textbf{totally split $\mathcal{k}$-algebra}: $\mathcal{A} = \mathcal{k}^r$ and $\sigma$ is defined by the cyclic left shift $(a_1,\dots, a_{r-1}, a_r) \mapsto (a_2,\dots, a_r, a_1)$. 
\item \textbf{Kummer extension}: $\mathcal{A} = \mathcal{k}(a)$ and $\sigma$ is defined by $a \mapsto \zeta a$, where $\zeta \in \mathcal{k}$ is a primitive $r$-th root of unity and $a^r \in \mathcal{k}$.
\item \textbf{Artin extension}: $\mathcal{A} = \mathcal{k}(a)$ and $\sigma$ is defined by $a \mapsto a + 1$, where $\operatorname{char}(\mathcal{k}) = r$ and $a^r - a \in \mathcal{k}$.
\end{itemize}
\subsection{Skew polynomial ring}
Let $\mathcal{A}$ be a $\langle \sigma \rangle$-Galois algebra. The \emph{skew polynomial ring} $\mathcal{A}[x,\sigma]$ is the ring whose underlying group is $\mathcal{A}[x]$ and the multiplication is defined by $x x^{d} = x^{d} x =  x^{d+1}, x a=\sigma(a) x$ for $d\in \mathbb{N}, a\in \mathcal{A}$. Since $r > 1$, $\mathcal{A}[x,\sigma]$ is a non-commutative graded $\mathcal{k}$-algebra. We denote by $\mathcal{A}[x,\sigma]_d$ (resp. $\mathcal{A}[x,\sigma]^d$) the subspace consisting of polynomials of degree $d$ (resp. at most $d$). For ease of reference, we record two basic properties of $\mathcal{A}[x,\sigma]$ below.
 \begin{lemma}\cite[Lemma~1.4]{2017Fast} \label{lem6}
     The map $\varphi: \mathcal{A}[x,\sigma]\to End_{\mathcal{k}}(\mathcal{A})$ sending $\sum_{i} a_{i} x^{i}$ to $\sum_{i}a_{i}\sigma^{i}$ is a surjective homomorphism of $\mathcal{k}$-algebras whose kernel is $(x^r - 1)$.
 \end{lemma}
\begin{lemma}\cite[Proposition~2.7]{2017Fast}    \label{lem7}
    Let $p_{1}, \dots,p_{m}\in \mathcal{k}[x]$ be pairwise coprime and let $p= \prod_{j=1}^m p_j$. The natural map:
    \begin{equation*}
        \mathcal{A}[x,\sigma]/p(x^{r})\to \mathcal{A}[x,\sigma]/p_{1}(x^{r})\times\cdots \mathcal{A}[x,\sigma]/p_{m}(x^{r})
    \end{equation*}
    is a $\mathcal{k}$-algebra isomorphism.
\end{lemma}
Let $\mathcal{k}^{r\times r}$ be the algebra of $r\times r$ matrices over $\mathcal{k}$. The following is a direct consequence of \cite[Proposition~30.6]{Knus98}.
 \begin{lemma}\label{lem:modular ring}
If $\mathcal{k}$ is algebraically closed, then $\mathcal{A}[x,\sigma]/(x^r - c) \simeq \mathcal{k}^{r\times r}$ for any nonzero $c \in \mathcal{k}$.
 \end{lemma}
\section{Lower bound}\label{sec:lower bound}
Since $\mathcal{A}[x,\sigma]$ is a $\mathcal{k}$-algebra, the multiplication map $\mu$ on $\mathcal{A}[x,\sigma]$ is $\mathcal{k}$-bilinear. For each $d\in \mathbb{N}$, we denote by $\mu_d$ the restriction of $\mu$ on polynomials of degree at most $d$. According to \eqref{eqn:rank<complexity}, we have $\rank_{\mathcal{k}} (\mu_d) \le C_{\mathcal{k}} (\mu_d)$. The goal of this section is to prove that $\rank_{\mathcal{k}} (\mu_d) =\Omega(\min(d,r)^{\omega-2}dr)$, which also provides a lower bound for $C_{\mathcal{k}} (\mu_d)$.
\begin{lemma} \label{thm9}
    If $d \ge r$ and $\mathcal{k}$ is algebraically closed, then $\rank_{\mathcal{k}} (\mu_d) \ge  d r^{\omega-1}$.
\end{lemma}
\begin{proof}
Without loss of generality, we may assume that $r | d$ and $\operatorname{char}(\mathcal{k}) \nmid (r/d)$. Let $\zeta\in \mathcal{k}$ be a primitive $d/r$-th root of unity. We denote $\mathcal{R} \coloneqq \mathcal{A}[x,\sigma]/(x^{d} - 1)$. Lemmas~\ref{lem7} and \ref{lem:modular ring} applied to $p(x) = x^{d/r} - 1$ and $p_j(x) = x - \zeta^j, 1 \le j \le d/r$ implies that 
\[
\mathcal{R} \simeq \bigoplus_{j=1}^{d/r} \mathcal{A}[x,\sigma]/(x^r - \zeta^j) \simeq \bigoplus_{j=1}^{d/r} \mathcal{k}^{r\times r}.
\]
Assume $\mu_{\mathcal{R}}$ is the multiplication map on $\mathcal{R}$. Let $\varphi_1:\mathcal{R} \to \mathcal{A}[x,\sigma]^{d}$ be the map defined by $\sum_{j=0}^{d-1} a_j x^j + (x^d - 1) \mapsto \sum_{j=0}^{d-1} a_j x^j$ and let $\varphi_2$ be the restriction to $\mathcal{A}[x,\sigma]^{2d}$ of the quotient map $\mathcal{A}[x,\sigma] \to \mathcal{R}$. Then we have $ \mu_\mathcal{R} = \varphi_2 \circ \mu_d   \circ (\varphi_1 \times \varphi_1)$. By \eqref{eqn:rank restriction} and Theorem~\ref{thm5}, we may conclude that $d r^{\omega-1} \le \rank_{\mathcal{k}} (\mu_d)$.
\end{proof}
\begin{lemma}\label{lem11}
If $\mathcal{k}$ is an infinite field, then there exist $g, p\in \mathcal{k}[t]$ of degrees $r$ and $r-1$ respectively, such that $\mathcal{A}[x,\sigma] \simeq \mathcal{k}\langle A,X\rangle/I$ as $\mathcal{k}$-algebras, where $\mathcal{k}\langle A, X\rangle$ is the non-commutative polynomial ring in variables $A, X$ over $\mathcal{k}$ and $I$ is the two-sided ideal generated by $g(A)$ and $X A - p(A)X$. 
\end{lemma}
\begin{proof}
By Theorem~\ref{thm:primitive element thm}, there exists $a\in \mathcal{A}$ such that $\mathcal{A} = \mathcal{k}[a]$. Thus one can find $p\in \mathcal{k}[t]$ of degree at most $(r-1)$ such that $p(a) = \sigma(a)$. Let $g \in \mathcal{k}[t]$ be the minimal polynomial of $a$. We claim that $g$ and $p$ are the desired polynomials. Indeed, the $\mathcal{k}$-linear map $q: \mathcal{k}\langle A,X\rangle\to \mathcal{A}[x,\sigma]$ induced by $A^i \to a^i, X^j \to x^j$ is surjective since $\mathcal{A} = \mathcal{k}[a]$. It is obvious that $I \subseteq  \ker(\psi')$ thus $\rho$ descends to $\psi': \mathcal{k}\langle A,X\rangle/I \to \mathcal{A}[x,\sigma]$. Next we define $\psi: \mathcal{A}[x,\sigma]\to \mathcal{k}\langle A,X\rangle/I$ by $\mathcal{k}$-linearly extending $\psi(a^i) = A^i, \psi(x^j) = X^j$. It is straightforward to verify that $\psi$ is a $\mathcal{k}$-algebra homomorphism and it is the inverse of $\psi'$.
\end{proof}
\begin{lemma}\label{thm12}
If $\mathcal{k}$ is an infinite field and $d \le r/3$, then $\rank_{\mathcal{k}}(\mu_d) \ge d^{\omega-1} r$.
\end{lemma}
\begin{proof}
By the isomorphism $\psi$ in the proof of Lemma~\ref{lem11}, we have $\rank_{\mathcal{k}} (\mu_d) = \rank_{\mathcal{k}} (T_d)$, where $T_d$ is the restriction to $\psi(\mathcal{A}[x,\sigma]^d) = \operatorname{span}_{\mathcal{k}} \lbrace A^j X^k:0\le j \le r-1, 0\le k \le d \rbrace \subseteq \mathcal{k}\langle A, X\rangle/I$ of the multiplication on $\mathcal{k}\langle A,X\rangle/I$. 

Since $\mathcal{k}$ is a field, the functor $\otimes_{\mathcal{k}} \mathcal{A}$ is exact. Therefore we have a short exact sequence of free $\mathcal{A}$-modules:
\begin{equation*}
    0\to I^{\mathcal{A}} \to \mathcal{k}\langle A, X\rangle^{\mathcal{A}} = \mathcal{A}\langle A, X\rangle\to (\mathcal{k}\langle A,X\rangle/I)^{\mathcal{A}} \to 0.
\end{equation*}
This induces an $\mathcal{A}$-module isomorphism $(\mathcal{k} \langle A,X\rangle/I)^{\mathcal{A}} \cong \mathcal{A} \langle A,X\rangle/I^{\mathcal{A}} $, which is in fact an $\mathcal{A}$-algebra isomorphism. 

Next we consider the $\mathcal{A}$-linear map $\rho: \mathcal{A}\langle A, X\rangle/I^{ \mathcal{A}} \to \mathcal{A}^{r\times r}$ induced by $A^i \to \alpha^i, X^j \to \beta^j$ where $0 \le i \le r-1$ and $j\in \mathbb{N}$, where $\alpha = \operatorname{diag}(a,\sigma(a), \dots,\sigma^{r-1}(a))$ and 
\begin{equation}\label{eqn:cyclic permutation matrix}
    \beta =\begin{bmatrix}
        0 &1 &0 &\cdots &0\\
        0 &0 &1 &\cdots &0\\
        \vdots &\vdots &\vdots &\vdots &\vdots\\
        0 &0 &0 &\cdots &1\\
        1 &0 &0 &\cdots &0
    \end{bmatrix}.
\end{equation}
We notice that $I^\mathcal{A}$ is the ideal of $\mathcal{A} \langle A, X\rangle$ generated by $g(A)$ and $X A-p(A)X$, where $g,p\in \mathcal{k}[t]$ are polynomials as in Lemma~\ref{lem11}. Moreover, matrices $\alpha, \beta$ satisfy $g(\alpha) = \beta \alpha - p(\alpha) \beta = 0$.
Thus $\rho$ is a $\mathcal{A}$-algebra homomorphism. Moreover, $\rho$ is an $\mathcal{A}$-module isomorphism from $(\psi(\mathcal{A}[x,\sigma]^d) )^{\mathcal{A}}$ to $\mathcal{A}^{r\times r}$. We denote the inverse of this isomorphism by $\rho':\mathcal{A}^{r\times r} \to (\psi(\mathcal{A}[x,\sigma]^{r-1}) )^{\mathcal{A}}$.

By construction, $\langle r,r,r \rangle_{\mathcal{A}}$ coincides with $\rho \circ T_d^{\mathcal{A}} \circ (\rho' \times \rho')$ on $\mathbb{U}^d \times \mathbb{U}^d$ where $\mathbb{U}^d =  \rho \left( ( \psi(\mathcal{A}[x,\sigma]^d) )^{\mathcal{A}}\right)$. By a direct calculation, $\mathbb{U}^d$ consists of $P = (P_{jk}) \in \mathcal{A}^{r\times r}$ such that $P_{jk} = 0$ if either $k - j \ge d-1$ or $-1 \ge k - j \ge d - r - 1$. Hence $\mathbb{U}^d \beta^{-\lceil d/2 \rceil} = \beta^{-\lceil d/2 \rceil} \mathbb{U}^d $ consists of $Q = (Q_{jk}) \in \mathcal{A}^{r\times r}$ such that $Q_{jk} = 0$ if either $k - j
 \ge \lceil d/2 \rceil$ or $-\lceil d/2 \rceil \ge k - j \ge  \lceil d/2 \rceil - r - 1$. In particular, $\mathbb{U}^d \beta^{-\lceil d/2 \rceil} = \beta^{-\lceil d/2 \rceil} \mathbb{U}^d $ contains all block diagonal matrices where each block is of size $\lceil d/4 \rceil \times \lceil d/4 \rceil$. Moreover, we observe that $PQ =\beta^{\lceil d/2 \rceil} (\beta^{-\lceil d/2 \rceil}P) (Q\beta^{-\lceil d/2 \rceil}) \beta^{\lceil d/2 \rceil}$. Thus we obtain
 \[
\langle \lceil d/4 \rceil, \lceil d/4 \rceil, \lceil d/4 \rceil \rangle^{\oplus \lfloor 4r/d \rfloor} \le \langle r,r,r \rangle_{\mathcal{A}}|_{\mathbb{U}^d \times \mathbb{U}^d} \le T^{\mathcal{A}}_d.
 \]
Theorem~\ref{thm5} and \eqref{eqn:rank extn} imply $
d^{\omega-1} r \le \rank_{\mathcal{A}} (T^{\mathcal{A}}_d) \le \rank_{\mathcal{k}} (T_d)$. 
\end{proof}
We notice that if $d = C r$ for some $1 > C > 1/3$, then we clearly have 
\[
\rank_{\mathcal{k}}(\mu_d) \ge \rank_{\mathcal{k}}(\mu_{\lfloor r/3 \rfloor}) \ge {\lfloor r/3 \rfloor}^{\omega - 1} r \ge \frac{1}{9} d^{\omega - 1} r.
\]
This leads to the corollary that follows.
\begin{corollary}\label{cor:d<r}
If $\mathcal{k}$ is an infinite field and $d < r$, then $\rank_{\mathcal{k}}(\mu_d) \ge \Omega(d^{\omega-1} r)$.
\end{corollary}

Finally, we are ready to establish the lower bound for the complexity of skew polynomial multiplication over an arbitrary field $\mathcal{k}$.
\begin{theorem}\label{thm:lower bound}
Let $\mathcal{k}$ be a field (not necessarily algebraically closed) and let $\mathcal{A}$ be a $\langle \sigma \rangle$-Galois algebra. We have 
\[
C_{\mathcal{k}} (\mu_d) \ge \rank_{\mathcal{k}}(\mu_d) = \Omega(d\min\{d,r\}^{\omega-2} r),
\]
where $\mu_d$ is the bilinear map of multiplying degree-$d$ skew polynomials and $C_{\mathcal{k}} (\mu_d)$ is the total complexity of $\mu_d$. 
\end{theorem}
\begin{proof}
The inequality follows from \eqref{eqn:rank<complexity}. Let $\overline{\mathcal{k}}$ be the algebraic closure of $\mathcal{k}$.  Then $\mathcal{A}[x,\sigma]^{\overline{\mathcal{k}}} = \mathcal{A}^{\overline{\mathcal{k}}}[x,\sigma^{\overline{\mathcal{k}}}]$ is the skew polynomial ring defined by the $\langle \sigma^{\overline{\mathcal{k}}} \rangle$-Galois algebra $\mathcal{A}^{\overline{\mathcal{k}}}$. Moreover, by \eqref{eqn:rank extn} we have $\rank_{\mathcal{k}}(\mu_d) \ge \rank_{\overline{\mathcal{k}}}(\mu_d^{\overline{\mathcal{k}}})$. Lastly, Lemma~\ref{thm9} and Corollary~\ref{cor:d<r} imply $\rank_{\overline{\mathcal{k}}}(\mu_d^{\overline{\mathcal{k}}}) = \Omega(d \min\{d,r\}^{\omega-2} r)$.
\end{proof}
As a consequence of Theorem~\ref{thm:lower bound} and the upper bound \eqref{eqn:2017bound}, we may conclude that when $d \ge r$, both $\rank_{\mathcal{k}}(\mu_d)$ and $C_{\mathcal{k}} (\mu_d)$ are completely determined, up to a log factor.
\begin{corollary}\label{cor:lower bound}
If $d \ge r$, then there exist constant numbers $C_1, C_2,\tau >0$ such that 
\[
C_1 d r^{\omega - 1} \le \rank_{\mathcal{k}}(\mu_d) \le  C_{\mathcal{k}} (\mu_d) \le C_2 d (\log d)^{\tau} r^{\omega - 1}.
\]
\end{corollary}
\section{Average bilinear complexity}\label{sec:average}
In this section, we discuss the average bilinear complexity of simultaneously multiplying skew polynomials. Let $N$ be a positive integer. We recall that the bilinear map of multiplying skew polynomials of degree at most $d$ is denoted by $\mu_d$. Thus the bilinear map that simultaneously multiplies $N$ pairs of skew polynomials of degree at most $d$ is $\mu_d^{\oplus N}$. We define \emph{the average bilinear complexity} of $\mu_d^{\oplus N}$ by $\operatorname{A-rank}_{\mathcal{k}}(N) \coloneqq \rank_{\mathcal{k}}(\mu_d^{\oplus N})/N$.

For each integer $2- r \le k \le r$, the \emph{$k-$th diagonal} of $P = (P_{ij}) \in \mathcal{k}^{r \times r}$ is the sequence $P^{(k)} \coloneqq (P_{i,i+k-1})_{i = \max\{1,2-k\}}^{\min\{r, r + 1 - k\}}$.
\begin{lemma}\label{lem15}
Let $m \ge 0$ and $2- r \le k_1, k_2 \le r - m$ be fixed integers. Suppose $P_1,P_2$ are $r\times r$ matrices such that $P_i^{(s)} = 0$ whenever $s < k_i$ or $s > k_i + m$, $i = 1,2$. Then one can compute $P_1 P_2$ by $\widetilde{O}( m^{\omega-1} r)$ arithmetic operations. Moreover, the bilinear complexity of multiplying such matrices is also $\widetilde{O}( m^{\omega-1} r)$.
\end{lemma}
\begin{proof}
Let $\beta\in \mathcal{k}^{r\times r}$ the the permutation matrix defined in \eqref{eqn:cyclic permutation matrix}. On the one hand, there are integers $n_1,n_2$ such that $\beta^{n_1}P_1$ and $P_2 \beta^{n_2}$ are of the form $Q + Q' $, where $Q$ is $\lceil m/2 \rceil$-banded and $Q' = \begin{bmatrix}
0 & 0 \\
E & 0
\end{bmatrix}$ for some $E\in \mathcal{A}^{\lceil m/2 \rceil \times \lceil m/2 \rceil}$. It is clear that we can multiply such matrices by $O(m^{\omega - 1} r )$ arithmetic operations. On the other hand, we have $P_1 P_2 =\beta^{-n_1} \left[ (\beta^{n_1}P_1) (P_2 \beta^{n_2})\right] \beta^{-n_2}$. Thus $P_1 P_2$ can be computed by $\widetilde{O}(m^{\omega - 1} r )$ operations as well. The upper bound for bilinear complexity can be obtained similarly.
\end{proof}

According to the proof of Lemma~\ref{thm12}, $\mu_d^{\mathcal{A}}$ is also the restriction of the multiplication of matrices of the same form as those in Lemma~\ref{lem15}. Thus we obtain the following corollary.
\begin{corollary}\label{coro13}
    $\rank_{\mathcal{A}}(\mu_d^{\mathcal{A}})= \widetilde{O}(d^{\omega-1} r)$.
\end{corollary}

\begin{proposition}\label{prop:average}
Let $\mathcal{k}$ be an infinite field and let $\mathcal{A}$ be an $r$-dimensional $\langle \sigma \rangle$-Galois algebra  over $\mathcal{k}$. For $d \ll r$ and $N = \Omega(r)$, we have $\operatorname{A-rank}_{\mathcal{k}}(N) = \widetilde{O}( d^{\omega - 1} r)$.
\end{proposition}
\begin{proof}
Since $\mathcal{A} = \mathcal{k}[a] \simeq \mathcal{k}[t]/(g(t))$, where $g$ is the minimal polynomial of $a$. By \cite[Exercise~15.6]{burgisser2013algebraic}, we have \[
\langle \lceil r/2 \rceil \rangle \le \mu_{\mathcal{A}} \le \langle 2r - 1 \rangle,
\] 
where $\mu_{\mathcal{A}}$ is the multiplication on $\mathcal{A}$ and $\langle m \rangle$ denotes the component-wise multiplication on $\mathcal{k}^{m}, m\in \mathbb{N}$. Thus we have 
\[
\langle \lceil r/2 \rceil \rangle \otimes_{\mathcal{k}} \mu_d \le \mu_{\mathcal{A}}\otimes_{\mathcal{k}} \mu_d
\] 
as $\mathcal{k}$-bilinear maps. Let $s \coloneqq \rank_{\mathcal{A}}(\mu_d^{\mathcal{A}})$. Then $\mu_d^{\mathcal{A}} \le \langle s \rangle^\mathcal{A}$ as $\mathcal{A}$-bilinear maps. We recall \cite[Subsection~15.3]{burgisser2013algebraic} that for any $\mathcal{k}$-bilinear map $T$, $T^\mathcal{A} = T \otimes_{\mathcal{k}} \mathcal{A}$ and as $\mathcal{k}$-bilinear maps. Moreover, for any $\mathcal{A}$-bilinear maps $S \le S'$, it also holds that $S \le S'$ as $\mathcal{k}$-bilinear maps. Therefore, we obtain
\[
\mu_d^{\oplus \lceil r/2 \rceil \rangle} = \langle \lceil r/2 \rceil \rangle \otimes_{\mathcal{k}} \mu_d\le \mu_{\mathcal{A}}\otimes_{\mathcal{k}}  \mu_d   \le  \mu_{\mathcal{A}}\otimes_{\mathcal{k}} \langle s \rangle \le \langle s (2r -1) \rangle
\]
as $\mathcal{k}$-bilinear maps, from which the desired upper bound of $\operatorname{A-rank}_{\mathcal{k}}(N)$ follows from Corollary~\ref{coro13}.
\end{proof}


\section{Quasi-optimal algorithms}\label{sec:algorithms}
In this section, we present quasi-optimal algorithms to compute the product of degree $d\ll r$ skew polynomials, for $\langle \sigma \rangle$-Galois algebras listed in Subsection~\ref{subsec:etale algebra}. As before, we assume without loss of generality that $d<r/3$. Moreover,
as discussed in \cite[Subsection~2.1.2]{GHS20}, we may assume that $\vert \mathcal{k} \vert>3r$. 
\subsection{Totally split algebra}\label{subsec:totally split algebra}
Let $\mathcal{A} = \mathcal{k}^r$ and let $\sigma$ be the cyclic left shift as in Subsection~\ref{subsec:etale algebra}. We denote by $\{e_i\}_{i=1}^r$ the canonical basis of $\mathcal{k}^r$. The map $\varphi$ in Lemma~\ref{lem6} can be viewed as a surjective $\mathcal{k}$-algebra homomorphism 
\begin{equation}\label{eqn:map split algbra}
\varphi: \mathcal{A}[x,\sigma] \to \mathcal{k}^{r\times r}
\end{equation} 
induced by $e_i \mapsto \alpha_i$ and $x \mapsto \beta$, where $\alpha_i$ is the matrix whose elements are all zero, except the $(i,i)$-th one, which is one and $\beta$ is the matrix defined in \eqref{eqn:cyclic permutation matrix}. We remark that given $P\in \mathcal{k}^{r\times r}$ and $m\in \mathbb{Z}$, both $\beta^m P$ and $P \beta^m$ are obtained by re-arranging elements of $P$.
We observe that there is a $\mathcal{k}$-linear map $\psi:\mathcal{k}^{r\times r} \to \mathcal{A}[x,\sigma]$ such that $\psi\circ\varphi= \id_{\mathcal{A}[x,\sigma]^{r-1}}$ where $\varphi$ is the map in \eqref{eqn:map split algbra}.
\begin{algorithm}[!hpbt]
        \SetAlgoLined
		\KwIn{ $f_1,f_2\in \mathcal{A}[x,\sigma]^d$}
		\KwOut{$f_1 f_2$}
		Compute $P_1 = \varphi(f_1)$ and $P_2 = \varphi(f_2)$\;
        Compute $P = P_1 P_2$\;
        Compute $f = \psi(P)$\;
        Return $f$.
        \caption{low degree skew polynomial multiplication for totally split algebra}
        \label{alg-1}
\end{algorithm}
 \begin{proposition}\label{prop16}
Let $\mathcal{k}$ be a field and let $\mathcal{A} = \mathcal{k}^r$ be a totally split extension of $\mathcal{k}$. For $d < r/3$, Algorithm~\ref{alg-1} computes the multiplication of elements in $\mathcal{A}[x,\sigma]^d$ by $\widetilde{O}(d^{\omega-1}r)$ arithmetic operations, where $\sigma$ is the automorphism of $\mathcal{A}$ induced by $(a_1,\dots, a_{r-1}, a_r) \mapsto (a_2,\dots, a_{r}, a_1)$.
 \end{proposition}
 \begin{proof}
 Since $\varphi$ is a $\mathcal{k}$-algebra homomorphism, we have $\psi ( \varphi(f_1) \varphi(f_2) ) = \psi\circ \varphi (f_1 f_2) = f_1 f_2$. Thus Algorithm~\ref{alg-1} indeed computes $f_1 f_2$.
As for the cost of Algorithm~\ref{alg-1}, we notice that for $f(x) = \sum\limits_{i=0}^{r-1}\sum\limits_{j=0}^{d}c_{i}x^{j} \in \mathcal{A}[x,\sigma]^d$, computing $\varphi(f)$ costs $O(d r)$ operations.
Next $\varphi(f)$ is an $r\times r$ matrix of the same form as those in Lemma~\ref{lem15}, thus the multiplication of such matrices costs $\widetilde{O}(d^{\omega-1} r )$ operations. Lastly, given $P = (P_{ij}) \in \mathcal{k}^{r\times r}$, $\psi(P) = \sum_{i=0}^{r-1} (\sum_{j=0}^{r-1} P_{ij} e_i ) x^{j-i \pmod{r}}$ by definition. Thus $\psi(P)$ costs no arithmetic operations and Algorithm~\ref{alg-1} costs $\widetilde{O}( d^{\omega-1} r )$ operations.
 \end{proof}
 \subsection{Kummer extension}\label{subsec:kummer}
Let $\mathcal{A} = \mathcal{k}(a)$ be a degree-$r$ Kummer extension of $\mathcal{k}$ and let $\sigma$ be defined as in Subsection~\ref{subsec:etale algebra}. By definition, there exists some $c\in \mathcal{k}$ such that the minimal polynomial of $a$ is $t^r - c$. We also pick and fix a $r$-th primitive root of unity $\zeta$. We denote by $I_{1}, I_2$ the two-sided ideals of $\mathcal{k}\langle X, A \rangle$ generated by $X A- \zeta A X$ and $A^{r}-c$, respectively. By the same argument as in the proof of Lemma~\ref{lem11}, there is a $\mathcal{k}$-algebra isomorphism 
\begin{equation}\label{eqn:kummer iso}
\psi: \mathcal{A}[x,\sigma] \to \mathcal{k}\langle X, A\rangle/(I_{1}+I_{2}) 
\end{equation}
defined by $x \mapsto X$ and $a\mapsto A$.
\begin{lemma}    \label{lem20}
Let $I_{3}$ be the two-sided ideal of $\mathcal{k}\langle X, A \rangle$ generated by $X^{r}-1$ and $A^{r}-1$. The $\mathcal{k}$-linear map $\varphi: \mathcal{k} \langle X, A\rangle/(I_{1}+I_{3})\to \mathcal{k}^{r \times r}$ induced by $A^i \mapsto \alpha^i, X^j \mapsto \beta^j$ is a $\mathcal{k}$-algebra isomorphism. Here $ \alpha = \operatorname{diag}(\zeta^k)_{k=0}^{r-1}$ and $\beta$ is the permutation matrix defined in \eqref{eqn:cyclic permutation matrix}. 
\end{lemma}
\begin{proof}
A direct computation implies that $\alpha \beta=\zeta \beta \alpha$, $\alpha^{r}= \I_r$ and $\beta^{r}= \I_r$. Hence $\varphi$ is well-defined and is a $\mathcal{k}$-algebra homomorphism. It is surjective since $\{\alpha^{i}\beta^{j}: 0 \le i, j \le r - 1\}$ is a $\mathcal{k}$-basis of $\mathcal{k}^{r\times r}$. Moreover, $\{A^{i}X^{j}: 0 \le i, j \le r - 1\}$ is a $\mathcal{k}$-basis of $\mathcal{k} \langle X,A\rangle/(I_{1}+I_{3})$, thus $\varphi$ is injective.
\end{proof}
Since $I_1 = (XA - \zeta A X)$ is homogeneous, $\mathcal{k} \langle X,A\rangle/I_{1}$ is bi-graded. We denote $\deg_X(F)$ (resp. $\deg_A(F)$) the degree of $X$ (resp. $A$) in $f\in \mathcal{k} \langle X,A\rangle$. We say that $F$ has bi-degree $(\deg_X(F),\deg_A(F))$. Given $d ,e \in \mathbb{N}$, we denote $\left( \mathcal{k}\langle X,A\rangle/I_{1} \right)^{d,e} \coloneqq \operatorname{span}_{\mathcal{k}} \lbrace X^{i}A^{j}: 0\le i \le d, 0\le j\le e \rbrace$. Let $\pi_3$ the natural quotient map $\mathcal{k} \langle X, A\rangle/I_{1} \to \mathcal{k} \langle X, A\rangle/ (I_{1}+I_{3})$.
\begin{algorithm}[!htbp]
        
        \SetAlgoLined
		
		\KwIn{ $F_1, F_2\in \left( \mathcal{k}\langle X,A \rangle/I_{1}\right)^{d,e}$}
		\KwOut{$F_1 F_2$}
        Compute $d_1 = \deg_{X}(F_1)$, $e_1 = \deg_{A}(F_1)$, $d_2 = \deg_{X}(F_2)$, $e_2 =\deg_{A}(F_2)$\;
		Compute $G_1 = \pi_3(F_1)$ and $G_2 = \pi_3(F_2)$\;
        Compute $M_1 = \varphi(G_1)$ and $M_2 = \varphi(G_2)$\; 
        Compute $M = M_1 M_2$\; \label{alg2:step4}
        Compute $G = \varphi^{-1}(M)$\; \label{alg-2:step5}
        Find $F \in \left( \mathcal{k} \langle X,A\rangle/I_{1}\right)^{d_1+ d_2, e_1 + e_2}$ such that $\pi_3(F)= G$\;\label{alg-2:step6}
        Return $F$.
        \caption{multiplication on $\left( \mathcal{k} \langle X,A\rangle/I_{1} \right)^{d,e}$}
        \label{alg-2}
	\end{algorithm}
\begin{lemma}\label{lem21}
Let $\mathcal{k}$ be a field containing a primitive $r$-th root of unity and let $d, e < r/3$ be positive integers. Algorithm~\ref{alg-2} computes the multiplication of two elements in $\left( \mathcal{k}\langle X, A\rangle/I_{1} \right)^{d,e}$ by $\widetilde{O}(d^{\omega-1}r)$ arithmetic operations. 
\end{lemma}
\begin{proof}
We first prove that $F = F_1 F_2$. By definition, $F_1 F_2$ has bi-degree $(d_1 + d_2, e_1 + e_2)$. Since $\pi_3$ and $\varphi$ are $\mathcal{k}$-algebra homomorphisms and $\varphi$ is an isomorphism, $\pi_3(F_1 F_2)=\varphi^{-1}(\varphi(\pi_3(F_1))\varphi(\pi_3(F_2))) = G$. Hence $F_1 F_2$ is a solution for Step~\ref{alg-2:step6} in Algorithm~\ref{alg-2}. Thus it suffices to prove that it is also unique. Indeed, if $\pi_3(F') =\pi_3(F) = G$ for some $F,F'\in \left( \mathcal{k} \langle X,A\rangle/I_{1}\right)^{d_1+ d_2, e_1 + e_2}$, then $F'-F\in (I_1 + I_{3})/I_1$.
Since a nonzero element in $(I_1 + I_{3})/I_1$ has bi degree $(l,m)$ where $\min\{l,m\} \ge r > \max\{3d, 3 e\} >\max\{d_1 + d_2, e_1 + e_2\}$, we obtain $F' = F$.

Next we analyze the complexity. Clearly, the cost of the first three steps of Algorithm~\ref{alg-2} is $O(dr)$. Since $F_1, F_2 \in \left(\mathcal{k}\langle X,A\rangle/I_{1}\right)^{d,r-1}$, we may deduce from the definition of $\varphi$ in Lemma~\ref{lem20} that matrices $M_1, M_2$ in Step~\ref{alg2:step4} are of the same form as those in Lemma~\ref{lem15}. Thus $M = M_1 M_2$ can be computed by $\widetilde{O}(d^{\omega-1}r)$ arithmetic operations.
Step~\ref{alg-2:step5} is equivalent to determining $\mu_{ij}\in \mathcal{k}$ such that $M =\sum\limits_{i=0}^{2d}\sum\limits_{j=0}^{2e}\mu_{ij} \alpha^{j} \beta^{i}$, where $\alpha$ and $\beta$ are matrices given in Lemma~\ref{lem20}. 
By definition of $\alpha$ and $\beta$, it is sufficient to find polynomials $f_1,\dots, f_{2d+1} \in \mathcal{k}[t]$ of degree at most $2e$ such that
\begin{equation}\label{lem:kummer:eqn:interpolation}
    f_{k}(\zeta^{l})= M_{l+1, l +k},\quad 0 \le l \le r-1, 1 \le k \le 2d+1.
\end{equation}
Here $l  +k$ is understood as $l  + k -r$ if $l + k > r$. 
Since $2e + 1 < r$, the interpolation problem \eqref{lem:kummer:eqn:interpolation} in general has no solution. However, the correctness of Algorithm~\ref{alg-2} ensures that \eqref{lem:kummer:eqn:interpolation} is solvable. Hence we can compute $f_{k}$'s by polynomial interpolation at $2e+1$ points $\zeta^i, 0 \le i \le 2e$. Each interpolation can be done by a fast Fourier transform, which costs $\widetilde{O}(e)$ arithmetic operations \cite[Theorem~2.6]{burgisser2013algebraic}. Thus, the cost of of Step~\ref{alg-2:step5} is $\widetilde{O}(de)$ arithmetic operations. 
It is obvious that Step~\ref{alg-2:step6} only costs $O(dr)$ operations. 

Lastly, since $d, e < r$, the above analysis implies that the total cost of Algorithm~\ref{alg-2} is $\widetilde{O}(d^{\omega - 1} r)$. 
\end{proof}
Let $\pi_2$ be the natural quotient map $\mathcal{k} \langle X, A\rangle/I_{1} \to \mathcal{k} \langle X, A\rangle/I_{1}+I_{2}$. We denote 
\[
\mathcal{A}[x,\sigma]^{d,e} \coloneqq \operatorname{span}_{\mathcal{k}} \lbrace x^{i}a^{j}: 0\le i \le d, 0\le j\le e \rbrace.
\] 
In particular, we have $\mathcal{A}[x,\sigma]^{d,r-1} = \mathcal{A}[x,\sigma]^{d}$.
\begin{algorithm}[h]
        \SetAlgoLined
		\KwIn{ $f_1, f_2\in \mathcal{k}[x,\sigma]^d$}
		\KwOut{$f_1 f_2$}
        Compute $H_1 = \psi(f_1)$ and $H_2 = \psi(f_2)$\; \label{alg-3:step1}
        Find $F_1,F_2\in \left( K[X,A]/I_{1}\right)^{d, r-1}$ such that $\pi_2(F_1)=H_1$ and $\pi_2(F_2)=H_2$\;
       Find $F_{11}, F_{12},F_{13}\in \left( K[X,A]/I_{1}\right)^{d, \lfloor r/3 \rfloor}$ such that $F_1 = F_{11}+ A^{\lceil r/3\rceil}F_{12}+A^{2\lceil r/3\rceil}F_{13}$\;
       Find $F_{21}, F_{22},F_{23}\in \left( K[X,A]/I_{1}\right)^{d, \lfloor r/3 \rfloor}$ such that $F_2 = F_{21}+ A^{\lceil r/3\rceil}F_{22}+A^{2\lceil r/3\rceil}F_{23}$\;\label{alg-3:step4}
        For $1\le i, j \le 3$, compute $G_{ij} = F_{1i} F_{2j}$ \Comment{Algorithm~\ref{alg-2}}\; \label{alg-3:step5}
		Compute $F  =\sum\limits_{i=1}^{3}\sum\limits_{j=1}^{3} A^{(i-1)\lceil r/3\rceil}G_{ij} A^{(j-1)\lceil r/3\rceil}$\;\label{alg-3:step6}
		Compute $f = \psi^{-1}(\pi_2(F))$ \; \label{alg-3:step7}
        Return $f$.
        \caption{low degree skew polynomial multiplication for Kummer extension}
           \label{alg-3}
	\end{algorithm}
\begin{proposition}\label{thm-22}
Let $\mathcal{k}$ be a field containing a primitive $r$-th root of unity $\zeta$ and let $\mathcal{A} = \mathcal{k}(a)$ be a degree-$r$ Kummer extension of $\mathcal{k}$. For $d < r/3$, Algorithm~\ref{alg-3} computes the multiplication of elements in $\mathcal{A}[x,\sigma]^d$ by $\widetilde{O}(d^{\omega-1}r)$ arithmetic operations, where $\sigma$ is the automorphism of $\mathcal{A}$ induced by $a \mapsto \zeta a$.
\end{proposition}
\begin{proof}
The correctness of Algorithm~\ref{alg-3} follows from the facts that $\pi_2$ is a $\mathcal{k}$-algebra homomorphism and that $\psi$ is a $\mathcal{k}$-algebra isomorphism. It is obvious that the cost of Step~\ref{alg-3:step5} in Algorithm~\ref{alg-3} is $\widetilde{O}(d^{\omega-1}r)$, while the total cost of other steps is $O(dr)$. Thus Algorithm~\ref{alg-3} costs $\widetilde{O}(d^{\omega-1}r)$ arithmetic operations.
\end{proof}
\subsection{Artin extension}\label{subsec:artin extension}
Assume that $\mathcal{k}$ is a field of characteristic $r > 0$. Let $\mathcal{A} = \mathcal{k}(a)$ be a degree-$r$ Artin extension of $\mathcal{k}$ and let $\sigma$ be defined as in Subsection~\ref{subsec:etale algebra}. By definition, the minimal polynomial of $a$ is $t^r - t - c \in \mathcal{k}[t]$ for some $c\in \mathcal{k}$. Suppose that $J_{1}, J_2$ are the two-sided ideals of $\mathcal{k}\langle X, A \rangle$ generated by $XA-(A+1)X$ and $A^{r}- A -c$ respectively. Similar to the case of Kummer extension discussed in Subsection~\ref{subsec:kummer}, we have a $\mathcal{k}$-algebra isomorphism 
\begin{equation}\label{eqn:artin iso}
\psi: \mathcal{A}[x,\sigma] \to \mathcal{k}\langle X, A\rangle/(J_{1}+J_{2}) 
\end{equation}
defined by $x \mapsto X$ and $a\mapsto A$.
As in Subsection~\ref{subsec:kummer}, given $d,e\in \mathbb{N}$, we use $\mathcal{A}[x,\sigma]^{d}$, $\mathcal{A}[x,\sigma]^{d,e}$ and $\left( \mathcal{k}\langle X,A\rangle/J_{1} \right)^{d,e}$ to denote the $\mathcal{k}$-subspaces spanned by low degree elements.

\begin{lemma}\label{lem26}
Let $J_{3}$ be the two-sided ideal generated by $X^{r}-1$ and $A^{r}-A$.
     The map $\varphi: \mathcal{k}\langle X,A\rangle/(J_{1}+J_{3})\to \mathcal{k}^{r\times r}$ defined by $A^i \to \alpha^i, X^j \to \beta^j$ is a $\mathcal{k}$-algebra isomorphism, where $\alpha = \operatorname{diag} (k)_{k=0}^{r-1}$ and $\beta$ is the permutation matrix defined in \eqref{eqn:cyclic permutation matrix}.
\end{lemma}
\begin{proof}
A direct computation shows that $\beta \alpha = (\alpha + 1) \beta$, $\beta^{r}= \I_r$ and $\alpha^{r}= \alpha$. This implies that $\varphi$ is a well-defined $\mathcal{k}$-algebra homomorphism. Moreover, $\varphi$ is bijective since $\{X^{i}A^{j}: 0 \le i, j \le r-1\}$ and $\{\beta^{i}\alpha^{j}: 0 \le i, j \le r-1\}$ are $\mathcal{k}$-bases of $\mathcal{k}\langle X, A \rangle$ and $\mathbb{k}^{r \times r}$, respectively.
\end{proof}
Let $\pi_3: \mathcal{k}\langle X, A\rangle/J_{1} \to \mathcal{k}\langle X, A\rangle/ (J_{1}+J_{3})$ be the natural quotient map. Given $F \in \mathcal{k} \langle X,A\rangle/J_{1}$, we denote by $(\deg_X(F),\deg_A(F))$ the bi-degree of $F$. By the same proof as that for Lemma~\ref{lem21}, we obtain the lemma that follows. 
\begin{algorithm}[h]
        \SetAlgoLined
		\KwIn{ $F_1 , F_2 \in \left(\mathcal{k}\langle X,A\rangle/J_{1}\right)^{d,e}$}
		\KwOut{$F_1 F_2$}
        Compute $d_1 = \deg_{X}(F_1)$, $e_1 = \deg_{A}(F_2)$, $d_2 = \deg_{X}(F_2)$, $e_2 = \deg_{A}(F_2)$\;
		Compute $G_1 = \pi_3(F_1)$ and $G_2 = \pi_3(F_2)$\;
        Compute $M_1 = \varphi(G_1)$ and $M_2= \varphi(G_2)$\;
        Compute $M = M_1 M_2$\;
        Compute $G = \varphi^{-1}(M)$\;
        Find $F \in \left( \mathcal{k}\langle X,A\rangle/J_{1} \right)^{d_1 + d_2, e_1 + e_2}$ such that $\pi_3(F)=\varphi^{-1}(G)$\;
        Return $F$.
        \caption{multiplication on $\left( \mathcal{k} \langle X,A\rangle/J_{1} \right)^{d,e}$}
        \label{alg-4}
	\end{algorithm}
\begin{lemma}
Let $\mathcal{k}$ be a field of characteristic $r > 0$ and let $d, e < r/3$ be positive integers. Algorithm~\ref{alg-4} computes the multiplication of two elements in $\left( \mathcal{k}\langle X, A\rangle/J_{1} \right)^{d,e}$ by $\widetilde{O}(d^{\omega-1}r)$ arithmetic operations. 
\label{lem27}
\end{lemma}

Next we discuss the cost of rewriting $\sum_{i=0}^{d} \sum_{j=0}^{r-1} \lambda_{ij} X^i A^j$ in $\left( \mathcal{k} \langle X, A \rangle/J_1 \right)^{d,r-1}$ in the form $\sum_{i=0}^{d}\sum_{j=0}^{r-1}\mu_{ij} A^{j}X^{i}$. We remark that in the case of Kummer extension discussed in Subsection~\ref{subsec:kummer}, it is easy to see that rewriting $\sum_{i=0}^{d} \sum_{j=0}^{r-1} \lambda_{ij} X^i A^j\in \left( \mathcal{k} \langle X, A \rangle/I_1 \right)^{d,r-1}$  as $\sum_{i=0}^{d}\sum_{j=0}^{r-1}\mu_{ij} A^{j}X^{i}$ costs $\widetilde{O}(dr)$ arithmetic operations, since $\mu_{ij} = \lambda_{ij} \zeta^{ij}$ by $X A = A X$. However, the rewriting procedure becomes more complicated in the case of Artin extension, since $XA = AX + X$ in $\mathcal{k} \langle X, A \rangle/J_1$.   
\begin{algorithm}[!hbpt]
        \SetAlgoLined
		\KwIn{$F = \sum_{i=0}^{d}\sum_{j=0}^{r-1}\lambda_{ij} X^{i}A^{j}\in \left( \mathcal{k} \langle X, A \rangle/J_1 \right)^{d,r-1}$ }
		\KwOut{$\mu_{ij}\in \mathcal{k}$ such that $F = \sum_{i=0}^{d}\sum_{j=0}^{r-1}\mu_{ij}A^{j} X^{i}$}
        Compute $M = \varphi (\pi_3(F))$ \Comment{fast polynomial evaluation}\; \label{alg-5:step1}
        Compute $\mu_{ij}$ from $M$ \Comment{fast polynomial interpolation}\;\label{alg-5:step2}
        Return $\mu_{ij}$.
        \caption{Exchange $X$ and $A$ in $\left( \mathcal{k} \langle X, A \rangle/J_1 \right)^{d,r-1}$}
        \label{alg-5}
	\end{algorithm}
 \begin{lemma}\label{lem28}
 Given $F =\sum_{i=0}^{d}\sum_{j=0}^{r-1}\lambda_{ij}X^{i} A^{j}$ in $\mathcal{k} \langle X, A \rangle/J_1$, Algorithm~\ref{alg-5} computes $\mu_{ij}\in \mathcal{k}, 0 \le i \le d, 0 \le j \le r-1$ such that $F=\sum_{i=0}^{d}\sum_{j=0}^{r-1}\mu_{ij}A^{j}X^{i}$ by $\widetilde{O}(dr)$ arithmetic operations. Similarly, if $F =\sum_{i=0}^{d}\sum_{j=0}^{r-1}\mu_{ij}z^{j}y^{i}\in  \left( \mathcal{k} \langle X, A \rangle/J_1 \right)^{d,r-1}$ is given, then one can rewrite $F$ as $\sum_{i=0}^{d}\sum_{j=0}^{r-1}\lambda_{ij}y^{i}z^{j}$ by $\widetilde{O}(rd)$ arithmetic operations as well. 
 \end{lemma}
 \begin{proof}
By symmetry, it is sufficient to prove the first part. We observe that $F =\sum_{i=0}^{d}\sum_{j=0}^{r-1}\lambda_{ij} X^{i} A^{j}=\sum_{i=0}^{d} X^{i}(\sum_{j=0}^{r-1}\lambda_{ij}A^{j})$. Thus, nonzero elements in $M = \varphi(\pi(_3(E))$ are simply evaluations of polynomials $f_0,\dots, f_{d} \in \mathcal{k}[t]$ at points $0,1\dots, r-1$. Here $f_{i}(t)=\sum_{j=0}^{r-1}\lambda_{ij}t^{j}, 0 \le i \le d$. By \cite[Corollary~3.20]{burgisser2013algebraic}, evaluations of each $f_i$ can be completed by $\widetilde{O}(r)$ operations. Hence the cost of Step~\ref{alg-5:step1} is $\widetilde{O}(dr)$.

By Lemma~\ref{lem26}, $\varphi$ is an isomorphism, thus $\mu_{ij}$'s can be determined by solving $M = \sum_{i=0}^d \sum_{j=0}^{r - 1} \mu_{ij} \alpha^j \beta^i$, where $\alpha,\beta$ are matrices defined in Lemma~\ref{lem26}. Therefore, $\mu_{ij}$'s can be obtained by interpolating $d+1$ polynomials of degree at most $r-1$ at points $0,\dots, r-1$. According to \cite[Corollary~3.22]{burgisser2013algebraic}, each polynomial interpolation can be done by $\widetilde{O}(dr)$ operations. The cost of Step~\ref{alg-5:step2} is $\widetilde{O}(dr)$ and this completes the proof.
 \end{proof}
  
Let $\pi_2: \mathcal{k}\langle X, A\rangle/J_{1} \to \mathcal{k}\langle X, A\rangle/(J_{1}+J_{2})$ be the natural projection .
By the same argument as in the proof of Proposition~\ref{thm-22}, we obtain the counterpart of Proposition~\ref{thm-22} for Artin extensions.
\begin{algorithm}[!hpbt]
        \SetAlgoLined
		\KwIn{ $f_1 ,f_2 \in \mathcal{A}[x,\sigma]^d$}
		\KwOut{$f_1 f_2$}
        Compute $G_1 = \psi(f_1)$ and $G_2 = \psi(f_2)$\;
        Find $F_1$ and $F_2$ in $\left(\mathcal{k}\langle X,A \rangle/J_{1}\right)^{d,r-1}$ such that $\pi_2(F_1)=G_1$ and $\pi_2(F_2)=G_2$\;
        Find $F_{11},F_{12}, F_{13}  \in \left(\mathcal{k}\langle X,A \rangle/J_{1}\right)^{d,\lfloor r/3 \rfloor}$ such that
         $F_1 =F_{11}+A^{\lceil r/3\rceil}F_{12}+A^{2\lceil r/3\rceil}F_{13} $ \Comment{ Algorithm~\ref{alg-5}}\;
          Find $F_{21},F_{22}, F_{23}  \in \left(\mathcal{k}\langle X,A \rangle/J_{1}\right)^{d,\lfloor r/3 \rfloor}$ such that
         $F_2 =F_{21}+ A^{\lceil r/3\rceil}F_{22}+ A^{2\lceil r/3\rceil}F_{23} $ \Comment{ Algorithm~\ref{alg-5}}\;
        For $1 \le i, j \le 3$, compute $G_{ij = }F_{1i}F_{2j}$ \Comment{Algorithm~\ref{alg-4}}\;
        For $1 \le i, j \le 3$, rewrite $A^{(i-1)\lceil r/3\rceil}G_{ij}$ as $\sum\limits_{s=0}^{d}\sum\limits_{t}\theta_{st}X^{s}A^{t}$ \Comment{Algorithm~\ref{alg-5}}\;
        Compute $F =\sum_{i=1}^{3}\sum_{j=1}^{3} A^{(i-1)\lceil r/3\rceil}G_{ij}A^{(j-1)\lceil r/3\rceil}$\;
        Compute $f = \psi^{-1}(\pi_2(F))$\;
        Return $f$.
        \caption{low degree skew polynomial multiplication for Artin extension}
        \label{alg-6}
	\end{algorithm}
\begin{proposition} \label{thm-23}
Let $\mathcal{k}$ be a field of characteristic $r$ and let $\mathcal{A} = \mathcal{k}(a)$ be an Artin extension of $\mathcal{k}$. For $d < r/3$, Algorithm~\ref{alg-6} computes the multiplication of elements in $\mathcal{A}[x,\sigma]^d$ by $\widetilde{O}(d^{\omega-1}r)$ arithmetic operations, where $\sigma$ is the automorphism of $\mathcal{A}$ induced by $a \mapsto  a+1$.
\end{proposition}
\subsection{Tower of Galois algebras}
This subsection is devoted to generalize Propositions~\ref{prop16}, \ref{thm-22} and \ref{thm-23}. Namely, we prove that the lower bound established in Section~\ref{sec:lower bound} is quasi-optimal for a $\langle \sigma \rangle$-Galois algebra $\mathcal{A}_2$ over $\mathcal{k}$, if there exists a tower $\mathcal{k} \subseteq \mathcal{A}_1 \subseteq \mathcal{A}_2$ such that either $\mathcal{A}_1/\mathcal{k} $ or $\mathcal{A}_2/\mathcal{A}_1$ is one the three types discussed in Subsections~\ref{subsec:totally split algebra}--\ref{subsec:artin extension}.

Let $\mathcal{k} \subseteq \mathcal{A}_1 \subseteq \mathcal{A}_2$ be a tower of finite dimensional \'{e}tale $\mathcal{k}$-algebras. We denote $r_1 \coloneqq \dim_{\mathcal{k}}\mathcal{A}_1$ and $r_2 \coloneqq \dim_{\mathcal{k}}\mathcal{A}_2$. Assume further that $\sigma$ is an automorphism of $\mathcal{A}_2$ such that $\sigma|_{\mathcal{A}_1}$ is also an automorphism of $\mathcal{A}_1$ and that both $\mathcal{A}_1$ and $\mathcal{A}_2$ are $\langle \sigma \rangle$-Galois algebra. 

For $i=1,2$, we denote by $\mu^{(i)}_d$ the $\mathcal{k}$-bilinear map of multiplying two elements in $\mathcal{A}_i[x,\sigma]^{d}$. Moreover, we notice that $\sigma^{r_1}$ is an automorphism of $\mathcal{A}_2$ and $\sigma^{r_1}|_{\mathcal{A}_1} = \id_{\mathcal{A}_1}$. Thus $\mathcal{A}_2[x,\sigma^{r_1}]^d$ is an algebra over $\mathcal{A}_1$. We denote the $\mathcal{A}_1$-bilinear map of multiplying two elements in $\mathcal{A}_2[x,\sigma^{r_1}]^d$ by $\mu_d^{(1,2)}$.

Moreover, we assume that each automorphism of $\mathcal{A}$ can be computed by $\widetilde{O}(r)$ arithmetic operations in $\mathcal{k}$. This is a direct consequence of the availability assumption for representational data in \cite[Assumption H]{GHS20}.
\begin{lemma}\label{lem30} 
Let $\mathcal{k},\mathcal{A}_1,\mathcal{A}_2, \sigma,\mu^{(1)}_d,\mu^{(2)}_d$ and $\mu^{(1,2)}_d$ be as above. Then we have $C_{\mathcal{k}}(\mu_d^{(2)})=\widetilde{O}\left(r_1^3 C_{\mathcal{A}_1}\left(\mu_{\lceil d/r_1\rceil}^{(1,2)}\right)\right)$. If moreover $\mathcal{A}_2 = \mathcal{A}_1[a]$ for some $a\in \mathcal{A}_1$, then $C_{\mathcal{k}}(\mu_d^{(2)}) = O\left( \left(\frac{r_2}{r_1} \right)^2 C_{\mathcal{k}} (\mu_d^{(1)}) \right)$.
\end{lemma}
\begin{proof}
Given $f_1,f_2\in \mathcal{A}_2[x,\sigma]^d$, we re-write \[
f_1=\sum_{k=0}^{r_{1}-1}x^{k} g_{k}(x^{r_{1}}),\quad f_2 =\sum_{l=0}^{r_{1}-1}h_{l}(x^{r_{1}})x^{l}
\]
for some $g_0,\dots, g_{r_1-1}, h_0,\dots, h_{r_1-1}\in \mathcal{A}_2[x,\sigma^{r_1}]^{\lceil d/r_1 \rceil}$. According to our assumption, one can find $g_j$'s and $h_j$'s by $\widetilde{O}(dr_{1})$ arithmetic operations, since the re-writing can be done by rearranging terms and computing automorphisms.

We notice that $f_1 f_2 =\sum_{k=0}^{r_{1}-1}\sum_{l=0}^{r_{1}-1}x^{k}g_{k}(x^{r_{1}})h_{l}(x^{r_{1}})x^{l}$.
Since $f_{k}, g_{l}\in \mathcal{A}_2[x,\sigma^{r_1}]^{\lceil d/r_1 \rceil}$, it costs $C_{\mathcal{A}_1} \left(\mu^{(1,2)}_{\lceil d/r_1 \rceil} \right) $ arithmetic operations in $\mathcal{A}_1$ to compute $f_k g_l$. Each operation in $\mathcal{A}_1$ has complexity $\widetilde{O}(r_1)$ over $\mathcal{k}$, thus the total complexity of computing $f_1 f_2$ is $\widetilde{O}\left( r_1^3 C_{\mathcal{A}_1} \left( \mu^{(1,2)}_{\lceil d/r_1 \rceil} \right) \right)$.

For the second part, we may re-write 
\[
f_1 =\sum_{k=0}^{r_{2}/r_{1}-1} a^{k} g_{k}(x),\quad f_2 =\sum_{l=0}^{r_2/r_1-1} h_{l}(x)a^{l}
\]
for some $g_{0},\dots, g_{r_2/r_1 - 1}, h_{0},\dots, h_{r_2/r_1 - 1} \in \mathcal{A}_1[x,\sigma]^{d}$. By the same argument as before, this re-writing again only costs $\widetilde{O}(r_{2}d/r_{1})$ operations in $\mathcal{k}$. We observe that each $f_{k}(x)g_{l}(x)$ is a product of two elements in $\mathcal{A}_1[x,\sigma]^d$. Thus it costs $C_{\mathcal{k}}(\mu_d^{(1)})$ arithmetic operations in $\mathcal{k}$. Since $f_1 f_2$ is the sum of $a^k g_k(x) h_l(x) a^l$, computing $f_1 f_2$ costs $O \left( \left(\frac{r_2}{r_1}\right)^{2} C_{\mathcal{k}}(\mu_d^{(1)}) \right)$.
\end{proof}

\begin{proposition}\label{coro31}
Assume that $\mathcal{A}_1$ is a field and $d = \Omega(r_1)$. If $\mathcal{A}_{2}/\mathcal{A}_{1}$ is a totally split algebra (resp. Kummer extension or Artin extension), then $C_{\mathcal{k}}(\mu_d^{(2)}) =\widetilde{O}\left(d^{\omega-1}{r_{2}}r_{1}^{3-\omega}\right)$.
\end{proposition}
\begin{proof}
A direct application of Propositions~\ref{prop16}, \ref{thm-22} and \ref{thm-23} to Lemma~\ref{lem30} leads to the desired  conclusion.
\end{proof}
Similarly, we also have the proposition that follows.
\begin{proposition}\label{coro33}
If $\mathcal{A}_{1}/\mathcal{k}$ is a totally split algebra (resp. Kummer extension or Artin extension) and $d = O(r_1)$, then $C_{\mathcal{k}}(\mu_2^{(d)})=\widetilde{O}(d^{\omega-1} r_{1}^{-1} r_{2}^{2})$.
\end{proposition}
\section*{Conclusion}
In this paper, we establish the inequality 
\[
\rank_{\mathcal{k}}(
\mu_d) \ge d\min\{d,r\}^{\omega-2} r,
\] 
where $\rank_{\mathcal{k}}(
\mu_d) $ is the bilinear complexity  of multiplying two skew polynomials in $\mathcal{A}[x,\sigma]$ of degree at most $d$ and $\mathcal{A}$ is a $\langle \sigma \rangle$-Galois algebra over $\mathcal{k}$ of dimension $r$. More importantly, this provides us a lower bound for the total complexity $C_{\mathcal{k}}(
\mu_d)$ of skew polynomial multiplication since $C_{\mathcal{k}}(
\mu_d) \ge \rank_{\mathcal{k}}(
\mu_d)$. We prove the quasi-optimality of this lower bound by presenting algorithms for special cases, including totally split algebras, Kummer extensions, Artin extensions and towers of these algebras. The complexity of our algorithms coincides with the conjectured upper bound in \cite{2017Fast}, which equals to our lower bound up to a log factor. We also prove that 
\[
\operatorname{A-rank}_{\mathcal{k}}(N) = \widetilde{O}( d^{\omega - 1} r).
\]
Here $\operatorname{A-rank}_{\mathcal{k}}(N)$ denotes the average of the bilinear complexity of simultaneously multiplying $N=\Omega(r)$ pairs of skew polynomials in $\mathcal{A}[x,\sigma]$ of degree at most $d \ll r$. 

For the future work, although our quasi-optimal lower bound together with the algorithm in \cite{2017Fast} completely determines (up to a log factor) $C_{\mathcal{k}} (\mu_d)$ for $d \ge r$, the upper bound of degree $d \ll r$ skew polynomial multiplication is still unknown in general. Namely, we do not know if there exists an algorithm of complexity $\widetilde{O}(d^{\omega - 1}r)$ that computes the multiplication of elements in $\mathcal{A}[x,\sigma]^d$ ($d\ll r$) for any $r$-dimensional $\langle \sigma \rangle$-Galois algebra $\mathcal{A}$. Moreover, results in this paper imply that if $d \ll r$, then 
\[
\operatorname{A-rank}_{\mathcal{k}}(N) = \widetilde{O}( d^{\omega - 1} r),\quad   d^{\omega-1} r \le \rank_{\mathcal{k}}(
\mu_d). 
\]
By definition, we also have $\operatorname{A-rank}_{\mathcal{k}}(N) \le \rank_{\mathcal{k}}(
\mu_d)$. However, it is unknown whether the equality holds.
\bibliographystyle{ACM-Reference-Format}
\bibliography{sample-base}
\end{document}